%% file: main_v2.tex
\documentclass[longbibliography]{revtex4-2}
\usepackage[normalem]{ulem}
\newcommand{\pij}{\langle ij \rangle}
\newcommand{\tauij}{\tau_{ij}}
\newcommand{\taubox}{\tau_{\Box}}
\newcommand{\rhobox}{\rho_{\Box}}
\newcommand{\ztaua}{Z_{\tau}(K_1, K_2)}
\newcommand{\ztaub}{Z_{\tau}(K, K_2)}
\newcommand{\tausum}{\sum_{\{\tau_{ij}=\pm 1\}}}
\newcommand{\ssum}{\sum_{\{S_i=\pm 1\}}}
\newcommand{\sigmasum}{\sum_{\{\sigma_i=\pm 1\}}}
\newcommand{\tausumsimp}{\sum_{\{\tau_{ij}\}}}
\newcommand{\ssumsimp}{\sum_{\{S_i\}}}
\newcommand{\sigmasumsimp}{\sum_{\{\sigma_i\}}}
\usepackage{braket}
\usepackage{amsmath}
\usepackage{amsthm}
\usepackage{amssymb}
\usepackage{bm}
\usepackage{latexsym}
\usepackage{graphicx}
\usepackage{xcolor}
\usepackage{ulem}
\usepackage[
	colorlinks=true,
	linkcolor=blue,
	citecolor=blue,
	urlcolor=blue,
	breaklinks
]{hyperref}
\newtheorem{theorem}{Theorem}
\newtheorem{lemma}{Lemma}

\begin{document}

\title{Analyticity of the energy in an Ising spin glass with correlated disorder}

\author{Hidetoshi Nishimori}
\affiliation{Institute of Innovative Research, Tokyo Institute of Technology, Yokohama 226-8503, Japan\\
Graduate School of Information Sciences, Tohoku University, Sendai 980-8579, Japan\\
RIKEN Interdisciplinary Theoretical and Mathematical Sciences Program (iTHEMS), Wako, Saitama 351-0198, Japan}

\begin{abstract}
The average energy of the Ising spin glass is known to have no singularity along a special line in the phase diagram although there exists a critical point on the line.  This result on the model with uncorrelated disorder is generalized to the case with correlated disorder. For a class of correlations in disorder that suppress frustration, we show that the average energy in a subspace of the phase diagram is expressed as the expectation value of a local gauge variable of the $Z_2$ gauge Higgs model, from which we prove that the average energy has no singularity although the subspace is likely to have a phase transition on it. Though it is difficult to obtain an explicit expression of the energy in contrast to the case of uncorrelated disorder, an exact closed-form expression of a physical quantity related to the energy is derived in three dimensions using a duality relation.  Identities and inequalities are proved for the specific heat and correlation functions. 
\end{abstract}

\maketitle

\section{Introduction}
The problem of spin glass is one of the most challenging topics in statistical physics \cite{nishimori2001,mezard1987}. Only a very limited number of exact solutions are known so far, among which the Parisi solution of the mean-field-type Sherrington-Kirkpartick model \cite{Sherrington1975} stands out as a distinguished achievement \cite{Parisi1979,talagrand2003}.  Another rare example is the exact solution for the average energy of the Ising spin glass in finite dimensions on a special line in the phase diagram \cite{Nishimori1981,Nishimori_1980,nishimori2001}. The resulting expression of the energy is written as a simple hyperbolic tangent of the inverse temperature without any singularity.  Since a multicritical point (transition point) lies on this special line in the phase diagram, it is highly non-trivial and counter-intuitive that the exact average energy is an analytic function.  Other physical quantities such as the specific heat and magnetic susceptibility are expected to have a singularity at the phase transition, but nothing is known so far exactly or rigorously on those quantities in finite-dimensional models except for the existence of a phase transition on the line as concluded from the existence of ferromagnetic order in the low-temperature side of the line \cite{Horiguchi1982} (see also \cite{Garban2021} for the continuous-variable case).

Those results concern the Edwards-Anderson model \cite{Edwards1975}, in which disorder of a given interaction bond distributes independently of disorder of other bonds.  In real spin glass materials, disorder exists mostly in sites, not in bonds \cite{Mydosh1993}. Edwards and Anderson assumed that  properties of spin glasses will be independent of the type of disorder and proposed their model with uncorrelated disorder.  One of the important differences between site and bond disorder is the existence or absence of correlation of disorder among nearby bond variables.  For example, if the position of a magnetic atom in a metal is affected by disorder, all interactions between this and other atoms are changed in  a correlated manner. It therefore makes sense to study how correlations in disorder would affect the system properties.

There have been several attempts to study the effects of correlation in disorder \cite{Hoyos2011,Bonzom2013,Cavaliere_2019,Munster2021}.  It has generally been observed that correlations, if weak, do not significantly modify the system properties qualitatively.  These studies exploit numerical methods or mean-field-type models, and it is desirable to establish analytical results in finite-dimensional systems.

We introduce a model with correlations of a specific type in disorder, which makes it possible to reduce the expression of the average energy to the expectation value of a local variable in the $Z_2$ gauge Higgs model \cite{Kogut1979}.  Based on this reduction, we prove that the average energy is an analytic function in a subspace of the phase diagram, a generalization of the special line in the case of uncorrelated disorder. It is difficult to derive a closed-form exact solution for the average energy when correlations are introduced.  Nevertheless, using self duality, we show in the case of three spatial dimensions that a simple formula holds for a physical quantity related to the energy. We also derive several non-trivial identities and inequalities for the specific heat and correlation functions, generalizing the results for uncorrelated disorder. Also discussed are more complex types of correlations in disorder.

In Sec. \ref{sec:main}, we define the model and analyze its properties. Conclusion is given in Sec. \ref{sec:conclusion}, and some technical details are described in Appendixes.

\section{Ising spin glass with correlated disorder}
\label{sec:main}
We first define the model of correlated disorder and then study the properties of the corresponding Ising spin glass, in particular the behavior of the average energy in a subspace of the phase diagram.

\subsection{Model of correlated disorder}
\label{subsec:model}
As mentioned in the Introduction, it is reasonable to study a spin-glass system with short-range correlations in disorder variables. It is likely that site disorder has a different degree of frustration than in the case of uncorrelated bond disorder, a prominent example of which is the Mattis model \cite{Mattis1976} without frustration but with site-dependent quenched disorder $\xi_i=\pm 1$,
\begin{align}
    H_{\rm Mattis}=-J\sum_{\pij}\xi_i\xi_j S_i S_j,
    \label{eq:Mattis}
\end{align}
where $S_i(\pm 1)$ is the Ising variable at site $i$, and $\pij$ runs over interacting spin pairs on a lattice with interaction strength $J$.
An example with site disorder but with frustration is the Hopfield model \cite{Hopfield1982,Amit1985a,Amit1985b,Amit1987},
\begin{align}
    H_{\rm Hopfield}=-\frac{1}{N}\sum_{i,j}J_{ij}S_iS_j,
\end{align}
where $N$ is the number of sites, and the interaction is composed of site-dependent quenched variables $\{\xi_i^{\mu}=\pm 1\}$,
\begin{align}
    J_{ij}=\sum_{\mu=1}^r \xi_i^{\mu}\xi_j^{\mu}.
\end{align}
This is a model with site-dependent disorder, in which correlation exists between disorder of different bonds $J_{ij}$ and $J_{jk}$ sharing a site index $j$.  When $r=1$, the model reduces to the Mattis model of Eq.~(\ref{eq:Mattis}) without frustration, and in the limit of $r\to\infty$, the model becomes equivalent to the Sherrington-Kirkpatrick model with uncorrelated bond disorder with strong frustration according to the central limit theorem. The Hopfield model thus interpolates the unfrustrated and fully frustrated cases by the parameter $r$.

These examples motivate us to introduce a system with correlated disorder in bond variables with a parameter to control the degree of frustration.  We therefore propose and analyze the following probability distribution of bond-assigned disorder variables $\{\tauij=\pm 1\}$ with control parameters $K_1$ and $K_2$,
\begin{align}
    P(\tau,K_1,K_2)=\frac{1}{Z_{\tau}(K_1, K_2)}\exp\Big(K_1\sum_{\pij}\tau_{ij}+K_2 \sum_{\Box}\tau_{\Box}\Big),
    \label{eq:Ptau}
\end{align}
where $Z_{\tau}(K_1, K_2)$ is the normalization factor (or the partition function of disorder variables)
\begin{align}
    Z_{\tau}(K_1, K_2)=\sum_{\{ \tau_{ij}=\pm 1\}} \exp\Big(K_1\sum_{\pij}\tau_{ij}+K_2 \sum_{\Box}\tau_{\Box}\Big),
    \label{eq:Ztau}
\end{align}
and $\tau_{\Box}$ is the product of four disorder variables around a unit plaquette denoted by $\Box$ on the $d$-dimensional hypercubic lattice, $\tau_{\Box}=\tau_{ij}\tau_{jk}\tau_{kl}\tau_{li}$
\footnote{
We can consider other possibilities. For example, on the two-dimensional triangular lattice, the product will run over three bond variables around a unit triangle, $\tauij\tau_{ji}\tau_{ki}$. We work on the hypercubic lattice in the present paper for simplicity.}. 
This is a generalization of the conventional $\pm J$ model $(K_2=0)$ to the case with correlation in disorder $(K_2>0)$. The above specific type of correlation was motivated by the following reasons.  (i) Reference \cite{Cavaliere_2019} studied the case of just the second term ($K_1=0$, $K_2>0$) and measured critical exponents  in three dimension by numerical simulations. We are interested in the interplay between the first ($K_1$) and second $(K_2)$ terms in the above more general model. (ii) The above form allows us to directly control the degree of frustration by the coefficient $K_2$.  (iii) The $K_2$ term is gauge invariant as will be discussed below, which is essential for our theory in the following sections to be applicable.
Notice that the exponent in the probability distribution Eq.~(\ref{eq:Ptau}) is of the same form as the action of the $Z_2$ gauge Higgs model \cite{Kogut1979}, and this analogy facilitates our analysis below. 

The Hamiltonian of the Ising spin glass is
\begin{align}
    H=-K\sum_{\pij} \tauij S_i S_j,
    \label{eq:Hamiltonian}
\end{align}
where $K$ is the dimensionless coupling constant (the inverse temperature $K=1/T$). We study the properties of this Hamiltonian with the quenched disorder variables $\{\tau_{ij}\}$ chosen from the probability distribution of Eq.~(\ref{eq:Ptau}). Possible generalization of the probability distribution will be discussed at the end of this section.

\subsection{Average energy in a subspace of the phase diagram}
\label{subsec:energy1}
\subsubsection{Expression of the average energy}
\label{subsubsec:Energy_expression}
We first show that the average energy can be expressed in a simple formula in terms of the $Z_2$ gauge Higgs model in a subspace of the phase diagram.

The average energy $E$ of the model defined in the preceding section is written as
\begin{align}
    E(K, K_1, K_2)=\frac{1}{\ztaua}\tausum e^{K_1\sum \tauij+K_2\sum \taubox}\langle H \rangle_K,
    \label{eq:Eoriginal}
\end{align}
where $H$ is the Hamiltonian of Eq.~(\ref{eq:Hamiltonian}) and the angular brackets denote the thermal average,
\begin{align}
    \langle H \rangle_K =-\frac{1}{Z_s(K)}\frac{\partial}{\partial K}Z_s(K)
    \label{eq:Hexpectation}
\end{align}
with the partition function of the Ising model,
\begin{align}
    Z_s(K)=\ssum e^{K\sum \tauij S_i S_j}.
    \label{eq:Zs}
\end{align}
Notice that the constant multiplicative factor $K$ is dropped for simplicity in Eq.~(\ref{eq:Hexpectation}) of the energy.

Following the standard prescription \cite{Nishimori1981,nishimori2001}, we apply the gauge transformation
\begin{align}
    S_i\to S_i\sigma_i,\, \tauij\to\tauij \sigma_i \sigma_j~(\sigma_i=\pm 1)~(\forall i,j)
\end{align}
to Eqs.~(\ref{eq:Eoriginal}) and (\ref{eq:Zs}). The thermal average $\langle H\rangle_K$ is gauge invariant, and only the term in the exponent of Eq.~(\ref{eq:Eoriginal}) with coefficient $K_1$  is affected.  We sum the resulting expression over all possible values of $\{\sigma_i\}$, and divide the expression by $2^N$, and find
\begin{align}
    &-E(K, K_1, K_2)\nonumber\\
    &=\frac{1}{2^N \ztaua}\tausum\sigmasum e^{K_1\sum\tauij\sigma_i\sigma_j+K_2\sum \taubox}\,
    \frac{\partial_K \ssumsimp e^{K\sum\tauij S_iS_j}}{\ssumsimp e^{K\sum\tauij S_iS_j}}.
    \label{eq:E_intermediate}
\end{align}
Since $\taubox$ is gauge invariant and $\{\sigma_i\}$ appears only in the term with coefficient $K_1$, the factor in the numerator $\sigmasumsimp e^{K_1\sum\tauij\sigma_i\sigma_j}$ cancels with the denominator $Z_s=\ssumsimp e^{K\sum\tauij S_iS_j}$ when $K=K_1$ and we are left with a simplified expression
\begin{align}
    -E(K,K, K_2)&=\frac{1}{2^N \ztaub}\tausumsimp e^{K_2\sum\taubox}\frac{\partial}{\partial K}\ssumsimp e^{K\sum\tauij S_i S_j}\nonumber\\
    &=\frac{1}{2^N \ztaub}\ssumsimp\tausumsimp \Big(\sum_{\pij} \tauij S_i S_j\Big) e^{K\sum\tauij S_i S_j+K_2\sum \taubox}\nonumber\\
    &=\frac{1}{\ztaub}\sum_{\pij}\Big(\tausumsimp \tauij \Big)\,e^{K\sum\tauij+K_2\sum \taubox},
    \label{eq:Egauge}
\end{align}
where we have applied the gauge transformation $\tauij \to \tauij S_i S_j$ to the second line to derive the third line.

In the case of uncorrelated distribution of disorder ($K_2=0$), the condition $K=K_1$ defines a line in the $K$-$K_1$ phase diagram, sometimes called the Nishimori line \cite{Georges1985,Doussal1988,Doussal1989,Iba1998,Gruzberg2001,Honecker2001,Nobre2001,Merz2002,Hasenbusch2008,Kitatani2009,Yamaguchi2010,Krzakala2011,Ohzeki2012,Sasagawa2020,Alberici2021,Alberici2021a,Garban2021}.  It is known that the low-temperature side of this line $(K=K_1\gg 1)$ lies in the ferromagnetic phase in two and higher dimensions \cite{Horiguchi1982,Dennis2002} and that the spin glass phase exists away from the line \cite{Nishimori1981,nishimori2001}. The line generally goes across a multicritical point where the paramagnetic, ferromagnetic and spin glass phases meet in three and higher dimensions \cite{nishimori2001,Ozeki1987}. Our present model with correlation in disorder has an additional parameter $K_2$ to control the degree of correlation, and the condition $K=K_1$ defines a subspace in the three-dimensional phase diagram drawn in terms of $K_1, K_2$, and $K$.

The expression of the average energy in Eq.~(\ref{eq:Egauge}) in the subspace $K=K_1$ can be interpreted as the expectation value of the local bond variable $\tauij$ summed over $\pij$ for the $Z_2$ gauge Higgs model \cite{Kogut1979}
\footnote{
Notice that the general expression of the energy of Eq.~(\ref{eq:Eoriginal}) is also regarded as the expectation value of a function of $\{\tauij\}$, $\langle H\rangle_K(\{\tauij\})$, with respect to the gauge Higgs probability weight $P_\tau$. However, this $\langle H\rangle_K(\{\tauij\})$ is a very complicated function of the bond variables $\{\tauij\}$ involving all types of products and summations of $\{\tauij\}$, which makes it hard to analyze its properties.}.
The condition $K=K_1$ has greatly simplified the expression of the average energy, leaving only the expectation value of a single $\tauij$ summed up over $\pij$.

\subsubsection{Analyticity of the average energy}
\label{subsubsec:analyticity}
The average energy of Eq.~(\ref{eq:Egauge}) is an analytic function of $K(=K_1)$ for sufficiently small $K_2 (> 0)$:
\begin{theorem}
\label{theorem}
The average energy $E(K, K, K_2)$ of Eq. (\ref{eq:Egauge}) is analytic in $K$ for sufficiently small $K_2$ in the thermodynamic limit.
\end{theorem}
This result is a known property of the $Z_2$ lattice gauge theory, see e.g. Ref.~\cite{Fradkin1979} for a theoretical analysis and Refs.~\cite{Creutz1980,Jongeward1980,CREUTZ1983,GENOVESE2003,Tupitsyn2010} for numerical studies. Since the argument presented in Ref.~\cite{Fradkin1979} is fairly sketchy, we provide a more formal proof in Appendix \ref{app:proof}.

When $K_2=0$ (uncorrelated disorder), the analyticity is trivial. Equation (\ref{eq:Egauge}) immediately gives a simple formula
\begin{align}
    E(K,K,0)=-N_B \tanh K,
    \label{eq:E_independent}
\end{align}
where $N_B$ is the number of bonds. Nevertheless, the line defined by $K=K_1$ in the $K$-$K_1$ phase diagram has a transition point on it since the low-temperature side ($K\gg 1)$ is in the ferromagnetic phase in two and higher dimensions \cite{Horiguchi1982} and the high-temperature side is trivially in the paramagnetic phase. Therefore the absence of singularity in $E(K,K,0)$ is highly non-trivial, a consequence of delicate balance of the probability weight $P(\tau,K_1, 0)$ and the partition function $Z_s(K)$ in Eq. (\ref{eq:E_intermediate}). Theorem \ref{theorem} claims that this non-trivial analyticity remains valid even after the introduction of correlation in disorder.  This may look rather trivial mathematically, but we believe it to be highly non-trivial physically. 

Although it is difficult to prove that the subspace $K=K_1$ has a transition point (critical point) on it when $K_2>0$, in contrast to the case $K_2=0$ where it is proven \cite{Horiguchi1982,Garban2021}, we expect that the physical properties of the system are unlikely to change dramatically by the introduction of small but finite $K_2$ as suggested by the continuation of analyticity of the average energy as in Theorem \ref{theorem}. An indirect indication of the stability of the low-temperature ($K=K_1\gg 1$) ferromagnetic phase after the introduction of $K_2$ is provided by small-$K_2$ perturbation of the distribution function of a single-bond variable defined by
\begin{align}
    P(l,K_1,K_2)=\tausum \delta_{l,\tauij}\,P(\tau,K_1, K_2) \quad (l=\pm 1),
\end{align}
where $\delta$ is Kronecker's delta. We define the probability for a bond to be positive (ferromagnetic) $p_+$ by
\begin{align}
    P(l,K_1,K_2)=p_+\,\delta_{l, 1}+(1-p_+)\, \delta_{l, -1}.
    \label{eq:pplusdef}
\end{align}
Then, as shown in Appendix \ref{appendix:pperturb}, $p_+$ is found by first-order perturbation in $\tanh K_2$ as
\begin{align}
    p_+=p\, \Big(1+2c (1-p)(2p-1)^3 \tanh K_2+\mathcal{O}(\tanh K_2)^2\Big),
    \label{eq:l-perturbed}
\end{align}
where $c$ is a positive constant and $p$ is the probability of $\tauij =1$ when $K_2=0$ (i.e., $p_+$ for $K_2=0$), defined through the ratio of probabilities for negative $(1-p)$ and positive $(p)$ values of $\tauij$: $e^{-K_1}/e^{K_1}=(1-p)/p$.  Equation (\ref{eq:l-perturbed}) suggests that the probability of a single bond to be positive (ferromagnetic) increases as we introduce small $K_2$.  This implies (but does not prove) that the low-temperature ferromagnetic phase becomes more stable by the introduction of $K_2>0$. This is intuitively reasonable because $K_2$ suppresses frustration, increasing stability of the ferromagnetic phase.  Since the existence of a paramagnetic phase for $K=K_1\ll 1$ is trivial, we may reasonably expect that there is a phase transition as a function of $K=K_1$ for $K_2>0$ in spite of the absence of singularity in the average energy.

The range of $K_2$ where the analyticity of average energy holds depends on the spatial dimension and the lattice structure. In the case of the two-dimensional square lattice, the $Z_2$ gauge Higgs model of Eq. (\ref{eq:Egauge}) is equivalent to the conventional ferromagnetic Ising model in a finite field on the square lattice according to the duality relation \cite{Wegner1971}.  Since the ferromagnetic Ising model in a field has no singularity, the energy of Eq.~(\ref{eq:Egauge}) is analytic for any positive value of $K_2$
\footnote{
This duality equivalence to the two-dimensional ferromagnetic Ising model in a field also implies the impossibility of the exact closed-form formula of the average energy.}
. In three dimensions, Monte Carlo simulation shows the range of analyticity to be $0\le K_2 < 0.6893$ on the cubic lattice \cite{Tupitsyn2010}.

\subsubsection{Identities for the average energy in three dimensions}
It is generally difficult to derive an explicit compact expression of the average energy for $K=K_1$ except for the case of $K_2=0$ as written in Eq.~(\ref{eq:E_independent}). Nevertheless, identities involving the average energy can be obtained in the case of the three-dimensional cubic lattice using the self duality of the model \cite{Wegner1971}.

Let us define the average energy per spin as
\begin{align}
    e(K_1,K_2)=\frac{1}{N}\, E(K_1, K_1,K_2)=-\frac{1}{N\ztaua}\tausum \Big(\sum_{\pij}\tauij\Big) \,e^{K_1\sum \tauij+K_2\sum \taubox}.
\end{align}
and the average plaquette energy per spin by
\begin{align}
    e_p(K_1,K_2)=-\frac{1}{N\ztaua}\tausum \Big(\sum_{\Box}\taubox\Big) \,e^{K_1\sum \tauij+K_2\sum \taubox}.
\end{align}
For these quantities, we can derive the following relation for a sufficiently large system on the cubic lattice where the boundary effects are sufficiently small, 
\begin{align}
    e(K_1,K_2)+3\tanh K_{1}=-\frac{1}{\sinh 2K_{1}}\, \big( e_p(K_1^*,K_2^*)+3\big),
    \label{eq:Eduality}
\end{align}
where the dual couplings $K_1^*$ and $K_2^*$ are defined as
\begin{align}
    \tanh K_1^*=e^{-2K_2},~\tanh K_2^*=e^{-2K_1}.
\end{align}
The proof is based on the duality relation of Wegner \cite{Wegner1971} for the $Z_2$ gauge Higgs model in three dimensions,
\begin{align}
    \frac{\ztaua}{2^{3N}(\cosh K_1\cosh K_2)^{3N}}=\frac{Z_{\tau}(K_1^*,K_2^*)}{e^{3N(K_1^*+K_2^*)}}.
\end{align}
 Self duality manifests itself as the same function $Z_{\tau}$ appearing on both sides. Taking the logarithmic derivative of both sides with respect to $K_1$ and dividing both sides by $N$, we obtain Eq.~(\ref{eq:Eduality}).

On the self-dual line $K_1=K_1^*$ (equivalently $K_2=K_2^*)$ in the $K_1$-$K_2$ plane, the above relation (\ref{eq:Eduality}) becomes
\begin{align}
    \sinh 2K_{1} \,\cdot e(K_1,K_2)+e_p(K_1,K_2)=-3\tanh K_{1}\, \sinh K_{1} -3.
\end{align}
This is a weighted sum of the averages of $\tauij$ and $\taubox$. A further simplification is realized when we impose another condition $K_1=K_2$ in addition to the self duality, which, together with $K_1=K_1^*$ and $K_2=K_2^*$, yields $\sinh 2K_1=\sinh 2K_2=1$, resulting in
\begin{align}
    e(K_c,K_c)+e_p(K_c,K_c)=-3\sqrt{2},
\end{align}
where $K_c=\frac{1}{2}\ln (\sqrt{2}+1)$ is the solution to $\sinh 2K_c=1$. This is the thermal average of the effective Hamiltonian $K_c\sum \tauij+K_c\sum \taubox$ of the $Z_2$ gauge Higgs model under the probability weight $P_\tau$ at a special point $(K_c,K_c)$ on the $K_1$-$K_2$ phase diagram.

\subsubsection{Specific heat and correlation functions}
\label{subsubsec:specific_heat}
Identities and inequalities can be proven for the specific heat and correlation functions in the subspace $K=K_1$, generalizing the results known in the case of uncorrelated disorder \cite{Nishimori1981,nishimori2001}.

The specific heat in the subspace $K=K_1$, $C(K, K, K_2)$, is bounded from above as follows.
\begin{align}
    &T^2 \, C(K,K,K_2)=-\left.\frac{\partial E(K,K_1,K_2)}{\partial K}\right|_{K_1=K} \nonumber\\
    &=\frac{1}{2^N \ztaua}\tausum\sigmasum e^{K_1\sum\tauij\sigma_i\sigma_j+K_2\sum\taubox}\nonumber\\
    &\left.\hspace{2cm}\times\left\{ \frac{\partial^2_K \ssumsimp e^{K\sum\tauij S_iS_j}}{\ssumsimp e^{K\sum\tauij S_iS_j}} -\left(\frac{\partial_K \ssumsimp e^{K\sum\tauij S_iS_j}}{\ssumsimp e^{K\sum\tauij S_iS_j}}\right)^2\right\}\right|_{K_1=K} \nonumber\\
    &\le \frac{1}{2^N\ztaua}\tausum e^{K_2\sum\taubox}\frac{\partial^2}{\partial K^2}\ssum e^{K\sum\tauij S_i S_j}\nonumber\\
    & \hspace{2cm}-\left(\frac{1}{\ztaua}\frac{\partial}{\partial K}\tausum e^{K\sum\tauij +K_2\sum\taubox}\right)^2 \nonumber\\
    &=-\frac{\partial}{\partial K}E(K, K, K_2) < \infty .
\end{align}
In deriving the third line from the second, we have replaced the average of the squared quantity $(\cdots )^2$ by the square of the average and have applied the condition $K_1=K$. The last inequality is based on the analyticity of the average energy for $K(=K_1)$. This inequality means that the specific heat does not diverge in the subspace $K=K_1$ for small $K_2> 0$, generalizing the result known for uncorrelated disorder \cite{Nishimori1981,nishimori2001}.

Similarly, we can verify the following relations on correlation functions using the method of gauge transformation as above, generalizing known relations \cite{Nishimori1981,nishimori2001}, 
\begin{align}
    &\left[ \langle S_i\rangle_K \right]_{K_1,K_2} =\left[ \langle S_i\rangle_{K_1} \langle S_i \rangle_{K}\right]_{K_1,K_2}
    \label{eq:correlation1}\\
   & \left[ \langle S_i\rangle_K\right]_{K_1,K_2}\le \left[ |\langle \sigma_i\rangle_{K_1}||\langle S_i \rangle_K|\right]_{K_1,K_2} \le \left[ |\langle \sigma_i\rangle_{K_1}|\right]_{K_1,K_2}
   \label{eq:correlation2}\\
  &  \left[ \frac{1}{\langle S_i \rangle_{K_{1}}}\right]_{K_1,K_2}=1.
\end{align}
Here the square brackets $[\cdots]_{K_1,K_2}$ denote the average over $\{\tauij\}$ by the probability $P(\tau,K_1,K_2)$. Notice that Eqs.~(\ref{eq:correlation1}) and (\ref{eq:correlation2}) hold for any values of $K_1, K_2$, and $K$, not just in the subspace $K=K_1$.  The first identity means that there is no spin glass phase, which has $[ \langle S_i\rangle_K ]_{K_1,K_2}=0$ and $[\langle S_i\rangle_K^2]_{K_1,K_2}>0$, when $K=K_1$. It is assumed that boundary spins are fixed to $+1$ such that the thermal average $\langle S_i \rangle_K$ is finite in the ferromagnetic phase.
\subsubsection{More general correlations}
\label{subsubsec:general_correlation}
Analyticity of the average energy holds for a more general class of correlation distributions as long as the invariance property under gauge transformation is satisfied.  A simple example is an additional term of the product of two neighboring plaquette variables,
\begin{align}
    P(\tau,K_1, K_2, K_3)=\frac{1}{Z_{\tau}(K_1, K_2, K_3)}\exp \Big(K_1\sum_{\pij}\tauij +K_2\sum_{\Box}\taubox+K_3\sum_{\Box_1\Box_2}\tau_{\Box_1}\tau_{\Box_2}\Big),
    \label{eq:Ptaugen}
\end{align}
where $\Box_1$ and $\Box_2$ are two neighboring plaquettes sharing a bond. It is straightforward to confirm that the proof of Theorem \ref{theorem} in Appendix \ref{app:proof} remains valid with minor changes after this modification with sufficiently small $K_2$ and $K_3$.

It is unclear whether or not Theorem \ref{theorem} holds in the case of infinitely many (in the thermodynamic limit) terms in the exponent of the probability distribution of bond variables, not just a finite number of terms as in Eq.~(\ref{eq:Ptaugen}). Nevertheless, the following example suggests that similar analyticity may be valid for a large class of probability distributions. Let us consider the probability distribution of correlated disorder,
\begin{align}
    P(\tau,K_1, K_0)=\frac{1}{Z_{\tau}(K_1, K_0)}\,e^{K_1\sum \tauij}\, Z_{\rm I} (K_0, \{\tauij\})=\frac{1}{Z_{\tau}(K_1, K_0)}\,e^{K_1\sum \tauij-F(K_0,\{\tauij\})},
    \label{eq:PIsing}
\end{align}
where $Z_{\rm I}(K_0, \{\tauij\})$ is the partition function of the Ising model on the same lattice with coupling $K_0$ and $F(K_0,\{\tauij\})$ is the corresponding free energy,
\begin{align}
Z_{\rm I}(K_0, \{\tauij\})=\sum_{\{\xi_{i}=\pm 1\}} e^{K_0\sum_{\pij}\tauij \xi_i\xi_j}=e^{-F(K_0,\{\tauij\})}\,.
\end{align}
The denominator $Z_{\tau}(K_1, K_0)$ is for normalization. The free energy can be expressed in terms of a cluster expansion (high-temperature expansion) for small $K_0$, generalizing  Eq.~(\ref{eq:Ptaugen}) to infinitely many (in the thermodynamic limit) gauge invariant terms,
\begin{align}
    -F(K_0,\{\tauij\})=K_1(K_0) \sum \taubox+K_2(K_0)\sum_{\rm connected} \taubox \taubox +K_3(K_0) \sum_{\rm connected} \taubox \taubox \taubox+\cdots,
\end{align}
where the summations run over connected clusters with positive coefficients $K_1, K_2, K_3,\cdots$ \cite{Friedli2017}.

The following result holds for this distribution of correlated disorder.
\begin{theorem}
\label{theorem2}
For sufficiently small $K_0$, the average energy is analytic in $K$  under the distribution function of Eq.~(\ref{eq:PIsing}) of correlated disorder.
\end{theorem}
\begin{proof}
Similarly to Eq.~(\ref{eq:E_intermediate}), the average energy after gauge transformation is
\begin{align}
    &-E(K,K_1,K_0)\nonumber\\
    &=\frac{1}{2^NZ_{\tau}(K_1,K_0)}\tausumsimp\sigmasumsimp e^{K_1\sum \tauij \sigma_i \sigma_j}\sum_{\{\xi_i\}}e^{K_0 \sum \tauij \xi_i \xi_j}\, \frac{\partial_K \ssumsimp e^{K\sum \tauij S_i S_j}}{\ssumsimp e^{K\sum \tauij S_i S_j}}.
\end{align}
When $K=K_1$, the summation over $\{\sigma_i\}$ cancels out with the partition function (summation over $\{S_i\}$) in the denominator,
\begin{align}
    -E(K,K,K_0)
    =\frac{1}{2^NZ_{\tau}(K,K_0)}\tausumsimp \sum_{\{\xi_i\}}e^{K_0 \sum \tauij \xi_i \xi_j}\,\frac{\partial}{\partial K} \ssumsimp e^{K\sum \tauij S_i S_j}\, .
\end{align}
We can carry out the summation over $\{\tauij\}$ at each $\pij$ independently,
\begin{align}
    \sum_{\{\tauij\}}\prod_{\pij}e^{\tauij (K_0 \xi_i\xi_j +KS_i S_i)} \propto e^{\tilde{K}\sum_{\pij}\xi_i\xi_j S_i S_j},
\end{align}
where $\tanh \tilde{K}=\tanh K_0\tanh K$. This relation can be verified by inserting all possible values of $\xi_i \xi_j=\pm 1$ and $S_i S_j =\pm 1$. By the gauge transformation $\xi_iS_i \to S_i$ for fixed $\xi_i$, we find that $E(K,K,K_0)$ is the energy of the ferromagnetic Ising model with uniform coupling $\tilde{K}$ without disorder.  Since the effective coupling $\tilde{K}$ runs from 0 to $K_0$ as $K$ changes from 0 to $\infty$,  the effective coupling $\tilde{K}$ stays below the critical point of the ferromagnetic Ising model for sufficiently small $K_0$. Therefore the energy is analytic in $K$.
\end{proof}

\section{Conclusion}
\label{sec:conclusion}
We have introduced correlations in disorder for the Ising spin glass problem in such a way that the degree of frustration is controlled.  We have analyzed the properties of the system in a subspace of the phase diagram in which the average energy is known to be analytic even across a transition point in the case of uncorrelated disorder. We have proved that the analyticity of the average energy is preserved under the introduction of correlations in disorder. This result remains valid for a few more complex types of correlations as long as correlations are gauge invariant.  Bounds on the specific heat and identities and inequalities have been derived, generalizing the results for uncorrelated disorder.

Our derivation of these results relies on gauge invariance of correlations.  It is difficult to analyze the case without gauge invariance such as a simple product of two neighboring bond variables $\tauij\tau_{jk}$.  It would be safe not to expect similar analyticity to hold in such a case because our result has been derived through a delicate balance of the term in the probability distribution and the partition function, which breaks down for gauge non-invariant probabilities.

For uncorrelated disorder, outstanding properties of the system in the special subspace have been found useful in applications in many fields including statistical inference \cite{Rujan1993,Sourlas1994,Murayama2000,Kabashima2000,Kabashima2000a,Montanari2001,Kabashima2001,Tanaka2002,Franz2002,Kabashima_2003,Macris2006,Macris2007,Decelle2011,Manoel_2013,Caltagirone2014,Xu_2014,Lesieur2015,Huang2016,Zdeborova2016,Huang_2017,Lesieur_2017,Kawamoto2018,Aubin_2019,Antenucci_2019,Kadmon_2019,Murayama2020,Vasiliy2020,Hou2020,Dall_Amico_2021,Kawaguchi2021,Arai2021,Alberici2021}, quantum error correction \cite{Dennis2002,Wang2002,Katzgraber2009,Katzgraber2010,Stace2010,Andrist2011,Bombin2012,Fujii2013,Fujii2014,Andrist2015,Iyer2015,Kubica2018,Kovalev2018,Li2019,Vuillot2019,zarei2019,Viyuela2019,Chubb2021}, quantum Hall effect \cite{Read2000}, and localization \cite{Senthil2000,Vodola2021}. The present result may stimulate further developments in these fields in addition to the spin glass theory itself.

\appendix

\section{Proof of Theorem \ref{theorem}}
\label{app:proof}
We prove by a cluster expansion that the following expression in Eq.~(\ref{eq:Egauge})
\begin{align}
    C_{01}(K_1, K_2)=\frac{\tausumsimp \tau_{01}\, e^{K_1\sum \tauij+K_2\sum (\taubox+1)}}{\tausumsimp  e^{K_1\sum \tauij+K_2\sum (\taubox+1)}}
    \label{eq:C01}
\end{align}
is analytic in $K_1$ for sufficiently small $K_2 (>0)$.

Notice that the exponent above has an additional term $K_2\sum 1$ both in the  numerator and the denominator compared to Eq.~(\ref{eq:Egauge}) but they cancel out to give the same quantity as in Eq.~(\ref{eq:Egauge}). This term is useful in the proof. The proof below is basically an adaptation of the theory in Ref.~\cite{Osterwalder1978} to the present context of the $Z_2$ lattice gauge model. See also Ref.~\cite{Fradkin1979}, where an abridged description of the theory is provided.

If we define $P(\tau,K_1)$ as $P(\tau,K_1,0)$, the above Eq.~(\ref{eq:C01}) is expressed as
\begin{align}
    C_{01}(K_1, K_2)=\frac{\tausumsimp \tau_{01} P(\tau,K_1) \,e^{K_2\sum (\taubox +1)}}{\tausumsimp  P(\tau,K_1)\, e^{K_2\sum (\taubox +1)}}.
    \label{eq:C01a}
\end{align}
To show that the small-$K_2$ expansion of $C_{01}(K_1,K_2)$ converges absolutely and uniformly and thus $C_{01}(K_1, K_2)$ is analytic in $K_1$, we introduce $\rhobox$ as
\begin{align}
    e^{K_2 (\taubox +1)}=1+\rhobox.
\end{align}
This $\rhobox$ is positive semi-definite and small for small $K_2$. The exponential factor in Eq. (\ref{eq:C01a}) is expanded as
\begin{align}
    e^{K_2 \sum_{\Box}(\taubox+1)}=\prod_{\Box}(1+\rhobox)=\sum_Q \prod_{\Box\in Q}\rhobox,
\end{align}
where $Q$ is the set of products of plaquettes.  Let us divide $Q$ into $Q_1$ and $Q_2$, where $Q_1$ is the set of connected products of plaquettes involving the bond $(01)$ and $Q_2$ is $Q\setminus Q_1$. Then
\begin{align}
    C_{01}(K_1,K_2)=\frac{1}{Z_{\tau}^+ (K_1,K_2)}\sum_{Q_1}\sum_{Q_2} \tausumsimp \tau_{01} P(\tau,K_1) \prod_{\Box\in Q_1}\rhobox \prod_{\Box\in Q_2}\rhobox,
\end{align}
where
\begin{align}
    Z_{\tau}^+(K_1, K_2)=\tausumsimp P(\tau,K_1)e^{K_2 \sum (\taubox+1)}=Z_{\tau}(K_1, K_2) e^{K_2 \sum 1}.
\end{align}
Since $Q_1$ and $Q_2$ are disjoint, 
\begin{align}
    C_{01}(K_1, K_2)=\frac{1}{Z_{\tau}^+ (K_1,K_2)}\sum_{Q_1}\sum_{\tau\in Q_1}\tau_{01} P(\tau\in Q_1, K_1) \Big(\prod_{\Box\in Q_1}\rhobox\Big) \cdot Z_{Q_2}(K_1, K_2),
\end{align}
where $Z_{Q_2}(K_1, K_2)$ is the partial partition function,
\begin{align}
    Z_{Q_2}(K_1, K_2)=\sum_{Q_2}\sum_{\tau\in Q_2} P(\tau\in Q_2,K_1) \prod_{\Box\in Q_2}\rhobox.
\end{align}
This partial partition function is bounded from above by $Z_{\tau}^+(K_1, K_2)$ because the coupling $K_2$ for $Q_1$ is set to 0 in $Z_{Q_2}$ in comparison with the full $Z_{\tau}$ and thus the summand is smaller.  See Remark 1 of Lemma 3.2 of Ref.~\cite{Osterwalder1978}.
We therefore replace $Z_{Q_2}(K_1,K_2)/Z_{\tau}^+(K_1,K_2)$ by 1 to have a bound
\begin{align}
    C_{01}(K_1, K_2)&\le \sum_{Q_1}\sum_{\tau\in Q_1}  P(\tau\in Q_1,K_1)\prod_{\Box \in Q_1}\rhobox \nonumber\\
    &\le \sum_k \sum_{|Q_1|=k}\,\sum_{\tau\in Q_1:|Q_1|=k}P(\tau\in Q_1,K_1)\prod_{\Box \in Q_1}\rhobox.
    \label{eq:Cbound1}
\end{align}
We have classified $Q_1$ by the number of elements in it.
Using the trivial inequality
\begin{align}
    \rhobox \le e^{2K_2}-1,
\end{align}
we can replace the average of $\prod_{\Box \in Q_1}\rhobox$ with respect to the weight $P(\tau\in Q_1,K_1)$ by the upper bound,
\begin{align}
   \sum_{\tau\in Q_1:|Q_1|=k} P(\tau\in Q_1,K_1)\prod_{\Box \in Q_1}\rhobox\le (e^{2K_2}-1)^k.
\end{align}
Equation (\ref{eq:Cbound1}) is then simplified as
\begin{align}
    C_{01}(K_1, K_2)\le \sum_k \sum_{|Q_1|=k} (e^{2K_2}-1)^k.
\end{align}
As stated in Lemma \ref{lemma} below, the number of possible connected graphs involving the bond $(01)$ is bounded as
\begin{align}
    N(k)\le c_1\cdot c_2{}^k,\label{eq:Nbound}
\end{align}
where $c_1$ and $c_2$ are positive constants. We therefore obtain
\begin{align}
    C_{01}(K_1, K_2)\le c_1\sum_k c_2{}^k \, (e^{2K_2}-1)^k.
\end{align}
The upper bound of $k$ runs to infinity in the thermodynamic limit, but the summation clearly converges for sufficiently small $K_2$. Uniform convergence follows this absolute convergence of an upper bound, proving analyticity of the left-hand side $C_{01}(K_1,K_2)$ as a function of $K_1$ if we regard $K_1$ as a complex number in the neighborhood of the origin.
\hfill
$\square$
\\
\noindent
{\em Remark.}
The proof remains valid when $\tau_{01}$ is replaced by the product of a finite number of bond variables.

\begin{lemma}
\label{lemma}
 The number of graphs in $Q_1$ is bounded as in Eq.~(\ref{eq:Nbound}).
\end{lemma}
This is Lemma 3.3 of Ref.~\cite{Osterwalder1978} and is a plaquette version of the well-established bound for the number of connected graphs with variables on the bonds (edges), not on plaquettes, as described, for example, in Exercise 5.3 of Ref.~\cite{Friedli2017}. It is straightforward to replace bonds with plaquettes to derive the above Lemma.

\section{Evaluating $p_+$ to first order in $\tanh K_2$}
\label{appendix:pperturb}
According to the definition of Eq.~(\ref{eq:pplusdef}), $p_+$ is written as
\begin{align}
    p_+=\frac{\tausumsimp \delta_{\tau_{ij},1}e^{K_1\sum \tauij}\prod_{\Box}(1+\taubox\tanh K_2)}{\tausumsimp e^{K_1\sum \tauij}\prod_{\Box}(1+\taubox\tanh K_2)}.
    \label{eq:pplus}
\end{align}
The denominator is expanded to first order in $\tanh K_2$ as
\begin{align}
    (\cosh K_1)^{N_B}+N_p \tanh K_2 \,(\cosh K_1)^{N_B-4}(\sinh K_1)^4,
    \label{eq:App2den}
\end{align}
where $N_B$ is the number of bonds and $N_p$ is for the number of plaquettes.  Similarly, the numerator is
\begin{align}
    &(\cosh K_1)^{N_B-1}\sum_{\tauij}\delta_{\tauij,1}e^{K_1\tauij}\nonumber\\
    &+(\cosh K_1)^{N_B-5}\tanh K_2 \,(\sinh K_1)^4\sum_{\tauij} \delta_{\tauij,1}e^{K_1\tauij} (N_B-a)\nonumber\\
    &+(\cosh K_1)^{N_B-4}\tanh K_2\, (\sinh K_1)^3 \sum_{\tauij} \delta_{\tauij,1}\tauij \, e^{K_1\tauij}\cdot a,
\end{align}
where $a$ the number of plaquettes to which a given single bond belongs.
Inserting these expansions to Eq.~(\ref{eq:pplus}) and leaving the leading order, we have
\begin{align}
    p_+=p\Big(1-a(2p-1)^4 \tanh K_2+a (2p-1)^3\tanh K_2 \Big),
\end{align}
where we used $(1-p)/p=e^{-2K_1}$.
This is Eq.~(\ref{eq:l-perturbed}).

\input{ref.bbl}
\end{document}

%% file: ref.bbl
%